\documentclass{llncs}
\usepackage{setspace}
\usepackage{amssymb,amsmath}
\usepackage{url}
\usepackage{graphicx}
\usepackage{float}
\usepackage{algorithmic}
\usepackage{algorithm}
\usepackage{color}
\usepackage{epsfig}
\usepackage{amsmath}
\usepackage{amssymb}

\floatstyle{boxed}
\restylefloat{figure}

\usepackage{epsfig}
\usepackage{amsmath}
\usepackage{amssymb}

\def\squarebox#1{\hbox to #1{\hfill\vbox to #1{\vfill}}}

\def\eod{\vrule height 6pt width 5pt depth 0pt}

\newcommand{\negA}{\vspace{-0.05in}}
\newcommand{\negB}{\vspace{-0.1in}}

\newcommand{\mysubsection}[1]{\negB\subsection{#1}\negA}

\newcommand{\myparagraph}[1]{\par\smallskip\par\noindent{\bf{}#1:~}}

\newcommand{\be}{\begin{equation}}
\newcommand{\ee}{\end{equation}}

\def \1{1 \!\! 1}

\newcommand{\ca}{{\cal A}}

\newcommand{\eps}{\varepsilon}

\newtheorem{myclaim}[section]{Claim}

\newcommand{\remove}[1]{}
\newcommand{\comment}[1] {}

\newtheorem{mylemma}[theorem]{Lemma}
\newtheorem{cl}[lemma]{Claim}

\newcommand{\cm}[1]{}

\renewcommand{\baselinestretch}{1.125}

\newcommand{\cI}{{\cal I}}

\newcommand{\cJ}{{\cal J}}
\newcommand{\cT}{{\cal T}}
\newcommand{\fsap}{\sc fsap}
\newcommand{\sap}{\sc sap}
\newcommand{\fbap}{\sc fbap}
\newcommand{\fba}{\sc fba}
\newcommand{\ba}{\sc ba}
\newcommand{\dsa}{\sc dsa}
\newcommand{\fsapss}{{\textsc {fsap}}}
\newcommand{\fbass}{{\textsc {fba}}}
\newcommand{\acn}{{\cal A}_{Narrow\_Color}}

\newcommand{\etal}{{\em et al.}}
\newcommand{\AVAIL}{{\tt AVAIL}}
\newcommand{\SOL}{{\tt SOL}}
\newcommand{\MinW}{{\tt Min}}
\newcommand{\MaxW}{{\tt Max}}

\def \ee   {\varepsilon}

\newlength{\tablength}
\newlength{\spacelength}
\newcommand{\tabstar}{\hspace*{\tablength}}
\newcommand{\spacestar}{\hspace*{\spacelength}}
\def\obeytabs{\catcode`\^^I=\active}
{\obeytabs\global\let^^I=\tabstar}
{\obeyspaces\global\let =\spacestar}
\newenvironment{display}{\begingroup\obeylines\obeyspaces\obeytabs}{\endgroup}
\newenvironment{prog}{\begin{display}\parskip0pt\sf}{\end{display}}

\title{
Flexible Resource Allocation for Clouds and 
All-Optical Networks
}
\author{Dmitriy Katz\inst{1} \and  Baruch Schieber\inst{1} \and
Hadas Shachnai\inst{2}\thanks{
This work was partly carried out
during a visit to DIMACS supported by the National Science Foundation under grant number CCF-1144502.}
}

\institute{IBM T.J. Watson Research Center, Yorktown Heights, NY~10598. \mbox{E-mail: {\tt \{katzrog,sbar\}@us.ibm.com}}
\and  Computer Science Department,
Technion, Haifa 3200003, Israel. \mbox{E-mail: {\tt hadas@cs.technion.ac.il}} }

\date{}
\begin{document}

\maketitle
\begin{abstract}
Motivated by the cloud computing paradigm, and by
key optimization problems
in all-optical networks,
we study two variants of the classic job interval scheduling problem, where 
a reusable resource is allocated to competing job intervals
in a {\em flexible} manner.
Each job, $J_i$, requires the use of up to $r_{max}(i)$ units of the resource, with a profit of $p_i \geq 1$ accrued for
each allocated unit. The goal is to feasibly schedule a subset of the jobs so as to maximize the total profit.
The resource can be allocated either in {\em contiguous} or {\em
non-contiguous} blocks.  These problems can be viewed
as flexible variants of the well known {\em storage allocation} and 
{\em bandwidth allocation} problems.

We show that the contiguous version is strongly NP-hard,
already for instances where all jobs have the same profit and the
same maximum resource requirement. For such instances, we derive the best possible positive result, 
namely, a polynomial time approximation scheme.
We further show that the contiguous variant admits a
$(\frac{5}{4}+\eps)$-approximation algorithm, for any fixed $\eps >0$, on instances whose job intervals form
a {\em proper} interval graph. At the heart of the algorithm lies a non-standard parameterization
of the approximation ratio itself, which is of independent interest.

For the non-contiguous case, we uncover an interesting relation  
to the {\em paging} problem that leads to
a simple $O(n \log n)$ algorithm for {\em uniform} profit instances of
$n$ jobs. The algorithm is 
easy to implement and is thus practical. 
\end{abstract}

%


\section{Introduction}
\label{sec:intro}
\subsection{Background and Motivation}
Interval scheduling is
one of the basic problems in the study of algorithms,
with
a wide range of applications
in computer science and in operations research (see, e.g., \cite{KLPS07}).
We focus on scheduling intervals with resource requirements.
In this model, we have a set of intervals (or, {\em activities}) competing for a reusable resource.\footnote{Throughout
the paper, we use the terms `intervals' and `activities'  interchangeably.} Each
activity utilizes a certain amount of the resource for the duration of its execution
and frees it upon completion. The problem is to find a feasible schedule of the activities
that satisfies certain constraints, including the requirement that
the total amount of resource allocated
simultaneously to activities never exceeds the amount of resource
available.

In this classic model,  two well-studied variants are
the storage allocation (see, e.g., \cite{LMV00,BBR13}), and the bandwidth allocation
 problems (see, e.g., \cite{BB+00,CMS07}).
In the {\em storage allocation problem} ({\sap}), each activity requires the allocation of a {\em contiguous}
block of the resource for the duration of its execution. Thus, the input is often viewed as a set of axis-parallel
rectangles; and the goal is to pack a maximum profit subset of rectangles into a horizontal strip of a given height,
by sliding the rectangles vertically but not horizontally.
When the resource can be allocated in {\em non-contiguous} blocks,
we have the {\em bandwidth allocation} {(\ba)} problem, where we only need to allocate
to each activity the required amount of the resource.

Scheduling problems of this ilk arise naturally in scenarios where activities require (non-contiguous or contiguous) 
portion of an available resource, with some revenue associated with the allocated amount.
In the cloud computing paradigm, the resource can be servers in a large computing cluster, storage capacity,
or bandwidth
allocated to time-critical jobs (see, e.g. \cite{JMNY11,MH11}).

In {\em flexgrid} all-optical networks, 
the resource is the available spectrum of light that
is divided into frequency intervals
of variable width, with a gap of unused frequencies between them (see, e.g., \cite{JT+09,G10}).
Several  high-speed signals connecting
different source-destination pairs may share a link, provided they are assigned disjoint sub-spectra
(see, e.g., in \cite{RSS09}).
Given a path network and a set of connection requests, represented by intervals, each associated with
a profit per allocated spectrum unit and a maximum bandwidth
requirement, we need to feasibly allocate 
frequencies to the requests, with the goal of maximizing the total
profit. In the {\fbap} variant
the sub-spectra allocated to each request need not be contiguous,
while in the {\fsap} variant each request requires a contiguous
frequency spectrum.
\remove{
Modern technologies used in scheduling jobs that require cloud services (see, e.g., \cite{JMNY11,MH11}),
or in spectrum assignment in
optical networks \cite{G10,VKRC12},
allow
{\em flexible} allocation of the available resources.
These technologies are attracting much
interest, due to their higher efficiency and better utilization
of compute and network resources. However, allocating the resources
becomes even more challenging, compared to
previous rigid technologies.
This motivates our study of the flexible variants of the above problems.
}

\subsection{Problem Statement}
We consider a variant of {\sap} where each interval
can be allocated any amount of the resource up to its maximum requirement,
with a profit accrued from each resource unit allocated to it.
 The goal is to schedule the intervals contiguously, subject to resource availability,
so as to maximize the total profit. We refer to this variant below as
the {\em flexible storage allocation problem} ({\fsap}).

 We also consider the
{\em flexible bandwidth allocation} ({\fbap})
problem, where each interval specifies an upper 
bound on the amount of the resource it can be allocated,
as well as the profit accrued from each allocated unit of the resource.
The goal is to determine the amount of resource allocated to each interval so as to maximize the total profit.

Formally,
in our general framework, the input consists of a set ${\cal J}$ of $n$ intervals. Each interval
$J_i \in {\cal J}$ requires the
utilization of a given, limited, resource. The amount of resource available, denoted by $W>0$, is fixed
over time. Each interval $J_i$ is defined by the following parameters.
\begin{enumerate}
\item[(1)]
A left endpoint, $s_i \geq 0$, and a right endpoint, $e_i \geq 0$. In
this case $J_i$ is associated with the half-open interval $[s_i, e_i)$
on the real line.
\item[(2)]
The amount of resource allocated each interval, $J_i$, which can take any
value up to the {\em maximum possible value} for $J_i$, given by
$r_{max}(i)$.
\item[(3)]
The profit $p_i \geq 1$ gained for each unit of the resource allocated to $J_i$.
\end{enumerate}

A {\em feasible} solution has to satisfy the following conditions. $(i)$ Each interval $J_i \in {\cJ}$ is allotted an amount of the
resource in its given range. $(ii)$ The total amount of the resource allocated at any time does not exceed the available amount,
$W$. In {\fbap}, we seek a feasible allocation which maximizes
 the total profit accrued by the intervals.
In {\fsap}, we add the requirement that the allocation to each
 interval is a {\em contiguous} block of the resource.

Given an algorithm ${\ca}$, let ${\ca}({\cJ}),OPT({\cJ})$ denote the
 profit of ${\ca}$ and an optimal solution for a problem instance ${\cJ}$, respectively.
 For $\rho \geq 1$, we say that
${\ca}$ is a $\rho$-approximation algorithm if, for any instance ${\cJ}$, $\frac{OPT({\cJ})}{{\ca}({\cJ})} \leq \rho$.

\comment{
\myparagraph{Applications}
Both the {\sap} and the {\ba} problem
naturally arise in scenarios where activities require (contiguous or non-contiguous) static portions of a resource. A
computer program
may require a contiguous range of storage space (e.g., memory allocation) for a specific time interval.
A task may require certain amount of (non-contiguous) bandwidth, or a contiguous set of frequencies. The resource
may be a banner, where each task is an advertisement that requires a contiguous portion of the banner.
Our study of {\fsap} and {\fbap} is primarily motivated from critical resource allocation problems arising in all-optical
networks. We give the detailed description in the Appendix.
}
\subsection{Our Contribution}
We derive both positive and negative results. 
On the positive side, we uncover an
interesting relation of {\fbap} to the classic {\em paging} problem that leads to
a simple $O(n \log n)$ algorithm for {\em uniform} profit instances
(see Section \ref{sec:non_con}).
Thus, we substantially improve the running time of the best known
algorithm for {\fbap} (due to \cite{SWZ13}) which uses flow techniques.
Our algorithm is easy to implement and is thus practical.

On the negative side, we show (in Section \ref{sec:fsap_ptas}) that {\fsap} is strongly NP-hard,
already for instances where all jobs have the same profit and the
same maximum resource requirement. For such instances, we derive the best possible positive result, namely, a PTAS.
We also present (in the Appendix) a $\frac{2k}{2k-1}$-approximation algorithm,
where $k = \lceil \frac{W}{\MaxW} \rceil$,
which is
of practical interest.
We further show (in Section \ref{sec:proper_fsap}) that {\fsap} admits a
$(\frac{5}{4}+\eps)$-approximation algorithm, for any fixed $\eps >0$, on instances whose job intervals form
a {\em proper} interval graph.

\myparagraph{Techniques}
Our Algorithm, Paging\_{\fba}, for the non-contiguous version of the problem, uses an interesting relation to the offline {\em paging} problem.
The key idea is to view the available resource as slots in fast
memory, and each job (interval) $J_i$ as
$r_{max}(i)$ pairs of requests for pages in the main memory.
Each pair of requests is associated with a distinct page: one request at $s_i$ and one at $e_i-1$.
We apply Belady's offline
paging algorithm that $-$ in case of a page fault $-$ evicts the page that is
requested furthest in the future (see Section \ref{sec:non_con}). If a page remains in the fast memory
between the two times it was requested, then the resource
that corresponds to its fast memory slot is allocated to the
corresponding job. In fact, Paging\_{\fba} solves the flexible bandwidth allocation problem optimally for more general instances, where each interval $J_i$ has also
a lower bound, $0 < r_{min}(i)$ on the amount of resource $J_i$ is allocated.\footnote{In obtaining all other results, we assume that $r_{min}(i)=0$ for $1 \leq i \leq n$.}

At the heart of our $(\frac{5}{4}+ \eps)$-approximation algorithm for proper instances
lies a non-standard parameterization
of the approximation ratio itself. Specifically, the algorithm uses a parameter $\beta \in \{0, 1 \}$ to guess the fraction of
total profit obtained by {\em wide} intervals, i.e., intervals with
high maximum resource requirement, in some optimal
solution. If the profit from these intervals is at least this fraction $\beta$
of the optimum for the given instance, such a high profit subset of wide intervals is found by the algorithm; else, the algorithm proceeds to find a high profit
subset of {\em narrow} and {\em wide} intervals, by solving an LP relaxation of a modified problem instance.
In solving this instance, we require that
the profit from {\em extra} units of the resource assigned to wide intervals (i.e., above certain threshold value)
is bounded by a $\beta$ fraction of the optimum (see Section \ref{sec:proper_fsap}). This
tighter constraint guarantees a small loss in profit when rounding the (fractional) solution for the LP.
The approximation ratio of $(\frac{5}{4}+\eps)$ is attained by optimizing on the value of $\beta$.
We believe this novel technique can lead to better approximation algorithms for other problems as well.

\comment{
In developing an approximation scheme for uniform {\fsap} instances, we utilizes some nice structural properties of optimal
solutions. In particular, given a uniform {\fsap} instance, we first show that
there exists an optimal solution for the (relaxed) {\fbap} instance, which satisfies the following:
any interval that properly contains another interval
is allocated a positive amount of the resource only if all the
intervals contained in it are
allocated their maximum requirements.
This property is at the heart of our algorithm, Paging\_{\fba} (see Section \ref{sec:non_con}).
For uniform instances, it implies the existence of an optimal solution for {\fbap} where each interval is allocated either
$\MaxW$ or $W \bmod\MaxW$ resource units. This solution can then be converted with
small decrease in profit into a valid solution for {\fsap}, when
$\MaxW$ is small relative to $W$.
Secondly, for uniform {\fsap} instances in which $\MaxW$ is large
relative to $W$,
we show the existence of almost optimal solutions having a certain {\em strip} structure.
For such instances, our scheme finds a solution having this structure.
}
\subsection{Related Work}
\label{sec:related}
The classic interval scheduling problem, where each interval requires all of the resource for its execution,
is solvable in $O(n \log n)$ time \cite{BB+00}.
\comment{
Storage allocation problems have been studied since the 1960's (see, e.g., \cite{K73}). The
traditional goal is to contiguously store a set of objects in minimum size memory.
This is the well-known {\em dynamic storage allocation} {(\dsa)} problem (see, e.g., \cite{BK+04}
and the references therein).
}
The {\em storage allocation problem} {(\sap)}
is NP-hard, since it includes Knapsack as a special case.
{\sap} was first studied in \cite{BB+00,LMV00}. Bar-Noy~\etal~\cite{BB+00}
presented an approximation algorithm that yields a ratio of $7$.
Chen~\etal ~\cite{CHT02} presented a polynomial time exact algorithm for the special case
where all resource requirements are multiples of $W/K$, for
some fixed  integer $K \geq 1$.
Bar-Yehuda~\etal~\cite{BBCR09} presented a randomized algorithm for {\sap} with ratio $2+\eps$, and a deterministic algorithm with ratio
$\frac{2e-1}{e-1}+ \eps < 2.582$.  The best known result is a deterministic $(2+\eps)$-approximation algorithm
due to \cite{SVZ14b}.

The {\em bandwidth allocation} ({\ba}) problem is known to be strongly NP-hard, already for uniform profits \cite{CWMX12}.
The results of Albers~\etal~\cite{AAK99} imply a constant factor approximation (where the constant is about $22$).
The ratio was improved to $3$ by Bar-Noy~\etal ~\cite{BB+00}.
Calinescu~\etal~ \cite{CCKR11} developed a randomized
approximation algorithm for {\ba} with expected performance ratio of $2 + \eps$, for every $\eps > 0$.
The best known result is an LP-based deterministic $(2 + \eps)$-approximation algorithm for {\ba}
due to Chekuri~\etal~ \cite{CMS07}.

Both {\ba} and {\sap} have been widely studied also in the {\em non-uniform} resource case, where the amount of available resource may change over time. In this setting,
{\ba} can be viewed as the {\em unsplittable flow problem (UFP)} on a path. The best known result is a $(2+\eps)$-approximation algorithm due to Anagnostopoulos~\etal ~\cite{AG+14}.
Batra~\etal \cite{BG+15} obtained approximation schemes for some spacial cases.\footnote{See also the recent results on {\em UFP with Bag constraints (BagUFP)} \cite{CC+14}.}
For {\sap} with non-uniform resource, the best known ratio is $2+\eps$, obtained by a randomized algorithm of M{\"{o}}mke and Wiese \cite{MW15}.

The {\em flexible} variants of {\sap} and {\ba}
were introduced by Shalom~\etal \cite{SWZ13}. The authors study instances where
each interval $i$ has a minimum and a maximum resource requirement, satisfying
$0 \leq r_{min}(i) < r_{max}(i) \leq W$, and the goal is to find a maximum profit schedule, such that
the amount of resource allocated to each interval $i$ is in $[r_{min}(i), r_{max}(i)]$.
 The authors show that {\fbap} can be optimally solved using flow
techniques.
The paper also presents a $\frac{4}{3}$-approximation
algorithm for {\fsap}
instances in which the input graph is proper, and
$r_{min}(i) \leq \lceil  \frac{r_{max}(i)}{2}\rceil$,
for all $1 \leq i \leq n$.
The paper \cite{SVZ14a}  shows NP-hardness
of {\fsap} instances where each interval has positive lower and upper bounds
on the amount of resource it can be allocated. The problem remains
difficult even if the bounds are identical for all activities, i.e.,
$r_{min}(i)=\MinW$ and $r_{max}(i)=\MaxW$, for all $i$, where $0 <\MinW <  \MaxW \leq W$.
The authors also show that {\fsap} is NP-hard for the subclass of instances where
$r_{min}(i)=0$ and $r_{max}(i)$ is arbitrary, for all $i$, and present
a $(2+\eps)$-approximation algorithm for such instances, for any fixed $\eps >0$.
We strengthen the hardness result of \cite{SVZ14a}, by showing that {\fsap} is strongly NP-hard
even if $r_{min}(i)=0$ and $r_{max}(i)=\MaxW$, for all $i$.

Finally, the paper \cite{SVZ14b} considers variants of {\fsap} and {\fbap} where
$0 \leq r_{min}(i) < r_{max}(i) \leq W$, and the goal is to feasibly schedule a {\em subset} $S$ of the intervals
of maximum total profit (namely, the amount of resource allocated to each interval $i \in S$ is
in $[r_{min}(i),r_{max}(i)]$). The paper presents a $3$-approximation algorithm for this version of {\fbap}, and a
$(3 +\eps)$-approximation for the corresponding version of {\fsap}, for any fixed $\eps >0$.

\section{Preliminaries}
\label{sec:prel}
We represent the input ${\cJ}$ as an interval graph, $G=(V,E)$, in which the set
of vertices, $V$, represents the $n$ jobs, and there is an edge $(v_i, v_j) \in E$ if the intervals
representing the jobs $J_i$, $J_j$ intersect. For simplicity, we interchangeably use $J_i$ to denote the $i$-th
job, and the interval representing the $i$-th job on the real-line.
We say that an input ${\cJ}$ is proper, if in the  corresponding interval graph $G=(V,E)$,
no interval $J_i$ is properly contained in another interval $J_j$, for all $1 \leq i,j \leq n$.

Throughout the paper, we use {\em coloring} terminology when referring to the assignment
of bandwidth to the jobs. Specifically, the amount of available resource, $W$, can be viewed as
the amount of available distinct colors. Thus, the demand of a job $J_i$ for (contiguous) allocation from the
resource, where the allocated amount is an integer in the range
$[0,r_{max}(i)]$,
can be satisfied by {\em coloring} $J_i$ with a (contiguous) set of colors, of size in the range
$[0,r_{max}(i)]$.

Let $C=\{1, 2, \ldots , W \}$ denote the set of available colors.
Recall that in a {\em contiguous} coloring, $c$, each interval $J_i$ is assigned
a {\em block} of $|c(J_i)|$ consecutive colors in $\{ 1, \ldots , W\}$.
In a {\em circular} contiguous coloring, $c$, we have the set of colors positioned consecutively on a circle.
Each interval $J_i \in {\cJ}$ is assigned a block of $|c(J_i)|$ consecutive colors on
the circle. Formally,
$J_i$ can be  assigned {\em any} consecutive sequence of $|c(J_i)|$ indices,
$\{ \ell, (\ell \bmod W)+1, \ldots ,  [(\ell+|c(J_i)|-2) \bmod W] +1 \}$, where $1 \leq \ell \leq W$.
Given a coloring of the intervals, $c : {\cJ} \rightarrow 2^{C}$, let $| c(J_i)|$ be the
number of colors assigned to $J_i$, then the total profit accrued from $c$ is
$\sum_{i=1}^n p_i | c(J_i)|$.

Let $S \subseteq {\cJ}$ be the subset of jobs $J_i$ for which
$|c(J_i)| \geq 1$ in a
(contiguous) coloring $C$ for the input graph $G$.
We call the subgraph of $G$ induced by $S$, denoted $G_S=(S,E_S)$, the
{\em support graph} of $S$.

\section{The Flexible Bandwidth Allocation Problem}
\label{sec:non_con}
In this section we study {\fbap}, the non-contiguous version of our
problem. We consider a generalized version of
{\fbap}, where each activity $i$ has also a lower bound $r_{min}(i)$ on
the amount of resource it is allocated.

Shalom~\etal \cite{SWZ13} showed that this generalized {\fbap} can be solved optimally by using flow techniques.
We show that in the special case where all jobs have the same
(unit) profit per allocated color (i.e., resource unit), the
problem can be solved by an efficient algorithm based on Belady's well
known  algorithm for offline paging \cite{Be66}.
From now on, assume that we have a feasible instance,
that is, there are
enough colors to allocate at least $r_{min}(i)$ colors to each job
$J_i$.

To gain some intuition, assume first that $r_{min}(i)=0$ for all
$i \in [1..n]$. We view the available colors as slots in fast
memory, and each job (interval) $J_i$ as
$r_{max}(i)$ pairs of requests for pages in the main memory.
Each pair of requests is associated with a distinct page: one request
at $s_i$
and one at $e_i-1$.
We now apply Belady's offline
paging algorithm:
if a page remains in the fast memory
between the two times it was requested, then the color
that corresponds to its fast memory slot is allocated to the
corresponding job.

When $r_{min}(i) >0$, we follow the same intuition while allocating
at least $r_{min}(i)$ colors to each $J_i$, to ensure
feasibility. We show below that the optimality of the paging
algorithm implies the optimality of the solution for our {\fbap} instance.

The algorithm is implemented iteratively, by
reassigning colors as follows.
The algorithm scans the left endpoints of the intervals,
from left to right.
When the algorithm scans $s_i$,
it first assigns $r_{min}(i)$ colors to
$J_i$ to ensure feasibility.
The algorithm starts by assigning the available colors. If there are less
than $r_{min}(i)$ colors available at $s_i$,
the algorithm examines the intervals
intersecting $J_i$ at $s_i$ in decreasing order of right endpoints,
and feasibly decreases the
number of colors assigned to these intervals and reassigns them to
$J_i$, until $J_i$ is allocated $r_{min}(i)$ colors.
The feasibility of the instance implies that so many colors can be
reassigned.

Next, the algorithm
allocates up to $r_{max}(i)-r_{min}(i)$
additional colors to interval $J_i$, to maximize profit.
If $r_{max}(i)-r_{min}(i)$ colors are available at $s_i$, then they
are assigned to $J_i$.
If so many colors are unavailable, and thus
$J_i$ is assigned less than $r_{max}(i)$ colors, then the algorithm
follows Belady's algorithm to potentially assign additional colors to
$J_i$. The algorithm examines the intervals
intersecting $J_i$ at $s_i$, and in case there are such intervals with
larger right endpoint than $e_i$, it feasibly decreases the
number of colors assigned to such intervals with the largest right
endpoints (furthest in the future),
and increases the number of colors assigned to $J_i$, up to $r_{max}(i)$.

When the algorithm scans $e_i$, the right endpoint of an interval
$J_i$, the colors assigned to $J_i$ are released and become available.
The pseudocode for Paging\_{\fba} is given in the Appendix (see
Algorithm \ref{alg:opt_fba}).


\begin{theorem}
\label{thm:opt_fba}
{\rm Paging}\_{\fba} is an optimal $O(n \log n)$ time algorithm,
 for any {\fbap} instance where $p_i=1$ for all $1 \leq i \leq n$.
\end{theorem}
\begin{proof}
Given an instance of {\fbap}, where $p_i=1$ for all $1 \leq i \leq n$,
define a respective multipaging problem instance, where the size of the
fast memory is $W$, as follows.
(The multipaging problem is a variant of the paging problem, where more
than one page is requested at the same time.)
For each job $J_i \in \cJ$ define two types of page requests:
{\em feasibility} requests and {\em profit} requests.
For each $s_i \leq t < e_i$ there are
$r_{min}(i)$ feasibility requests for the same $r_{min}(i)$ pages.
In addition there are $r_{max}(i)-r_{min}(i)$ pairs of requests for
$r_{max}(i)-r_{min}(i)$ distinct pages. The
first request of each pair is at $s_i$ and the second at $e_i-1$.
It is not difficult to see that Belady's algorithm is optimal also for
the multipaging problem \cite{Lib98}.
Consider an implementation of Belady's algorithm for this instance
where the feasibility requests are always processed before the profit
requests. In this case it is easy to see that none of the pages that
are requested in the feasibility requests will be evicted before the
last time in which they are requested.
Note that Paging\_{\fba} implements this variant of Belady's
algorithm, where $\AVAIL$ corresponds to the available memory slots.
The optimality of Belady's algorithm guarantees that the number of
page faults is minimized. The
page faults whenever a new page is requested cannot be avoided. (Note
that there are $\sum_{i=1}^n r_{max}(i)$ such faults.)
Additional page faults may occur when the second occurrence of a
page in the pairs of the profit requests needs to be accessed.
Minimizing the number of such page faults is equivalent to maximizing
the number of pairs of profit requests for which the requested page is
in memory throughout the interval $[s_i,e_i)$, that is, maximizing the
profit of allocated resources in the {\fbap} instance.
Paging\_{\fba} can be implemented in $O(n \log n)$ time, by noting that
the total number of color reassignments is linear, and
that
once we sort the intervals by left endpoints, the `active' intervals
can be stored in a priority queue by right endpoints to implement the
reassignments.
\hfill \eod
\end{proof}

\section{Approximating Flexible Storage Allocation}
\label{sec:proper_fsap}
In this section we consider the flexible storage allocation problem.
We focus below on {\fsap} instances in which the jobs form a {\em proper} interval graph, and give an approximation algorithm
that yields a ratio of $(\frac{5}{4} + \eps)$ to optimal.
Our Algorithm, Proper\_{\fsap}, uses the parameters $\eps >0$ and $\beta \in (0,1)$ (to be determined).
Initially, Proper\_{\fsap}  guesses the value of an optimal solution
${OPT}_{\fsapss}({\cJ})$.
(The guessing is done by binary search and it is straightforward to
verify by the results of Proper\_{\fsap} whether the guess is correct, or whether it is under or above the optimal value.)

Let ${\cJ}_{wide}$ denote the set of {\em wide} intervals $J_i$ for which
$r_{max}(i) \ge \eps W$. Let ${\cJ}_{narrow}$ denote the complement
set of {\em narrow} intervals.
If the profit from the {\em wide} intervals that are actually assigned at least
$\eps W$ colors is {\em large}, namely,
at least $\beta \cdot {OPT}_{\fsapss}({\cJ})$, then
such a high profit subset of intervals is found and returned by the algorithm as the solution. Otherwise,
Proper\_{\fsap} calls Algorithm ${\acn}$ that finds a solution of
high profit accrued from both {\em narrow} and {\em wide} intervals.
The pseudocode for Proper\_{\fsap} follows.

\begin{algorithm}[h]
\caption{Proper\_{\fsap}$({\cJ}, {\bar r_{max}}, {\bar p},W, \eps, \beta)$ \label{alg:proper_fsap}}
  \begin{algorithmic}[1]
   \STATE Guess the value of ${OPT}_{\fsapss}({\cJ})$, the optimal solution of {\fsap} on the input ${\cJ}$.
   \STATE Find in ${\cJ}_{wide}$ a solution $S$ of {\fsap} with maximum
   total profit among all solutions in which intervals that are assigned colors are assigned at least
   $\lceil \eps W \rceil$ colors.\label{alg:find_wide_solution}
 \IF {$P(S) < \beta OPT_{\fsapss}({\cJ})$}
      \STATE Let $S$ be the solution output by ${\acn}({\cJ}, {\bar r_{max}}, {\bar p},W, \eps, \beta {OPT}_{\fsapss}({\cJ}))$
\ENDIF
\STATE Return $S$ and the respective contiguous coloring
\end{algorithmic}
\end{algorithm}

We now describe Algorithm ${\acn}$ that finds an approximate
solution for {\fsap} in case the extra profit of the {\em wide}
intervals $-$ above the profit of their first $\lceil \eps W \rceil$ assigned colors $-$
is bounded by $\beta$ fraction of the optimal solution.

First, ${\acn}$ solves a linear program $LP_{\fbass}$
that finds a (fractional) maximum profit solution of the
{\fbap} problem on the set ${\cJ}$,
in which the number of colors used is no more than $(1-\eps)W$.
Note that, according to our guess, the value of the solution is at least
$(1-\eps) OPT_{\fsapss}({\cJ})$. This is since the value of the
optimal solution for $LP_{\fbass}$ when all $W$ colors are used is at least
$OPT_{\fsapss}({\cJ})$.
The solution needs to satisfy an upper bound on
the extra profit accrued from {\em wide}
intervals that are assigned more than $\eps W$ colors.

Next, this solution is rounded to an integral solution of the {\fbap}
instance, of value at least $(1-2\eps){OPT}_{\fsapss}({\cJ})$.
${\acn}$ proceeds by converting the resulting (non-contiguous)
coloring to a contiguous circular coloring with the same profit.
Finally, this coloring is converted to a valid (non-circular) coloring.
In this part, ${\acn}$ searches for the `best' index for `cutting' the circular coloring. This is done by examining  a polynomial number of integral points,
$\ell \in [1, (1-\eps)W]$, and calculating in each the loss in profit due to eliminating at most half of the colors for each {\em wide} interval whose (contiguous) color set includes
color $\ell$. The algorithm `cuts the circle' in the point $\ell$ which causes the smallest harm to the total profit. For each {\em wide} interval $J_i$  crossing $\ell$, we assign the largest among
its first block of colors (whose last color is $\ell$), or the second block of colors (which starts at $\ell \bmod \lfloor (1-\eps)W \rfloor +1$). For each {\em narrow} interval that included $\ell$, we assign the same number of new colors in the range
$[(1-\eps )W, W]$.
We give the pseudocode of ${\acn}$ in Algorithm \ref{alg:color_narrow}
in the Appendix.

\subsection{Analysis of Proper\_{\fsap}}
Our main result is the following.
\begin{theorem}
\label{thm:proper_fsap_ratio}
Proper\_{\fsap} is a $(\frac{5}{4}+\eps)$-approximation algorithm for any instance of {\fsap} in which
the input graph is proper.
\end{theorem}

We prove the theorem using the following results.
First, we consider the case where
there exists an optimal solution of {\fsap} in which
the profit from the intervals in ${\cJ}_{wide}$ that are assigned at least
$\eps W$ colors is
at least $\beta {OPT}_{\fsapss}({\cJ})$.
\begin{mylemma}
\label{lemma:profitable_wide}
If there exists an optimal solution of {\fsap} for ${\cJ}$ in which
the profit from the intervals in ${\cJ}_{wide}$ that are assigned at least
$\eps W$ colors is
at least $\beta {OPT}_{\fsapss}({\cJ})$, then
a solution of profit at least $\beta {OPT}_{\fsapss}({\cJ})$
can be found in polynomial time.
\end{mylemma}

\begin{proof}
We consider only the intervals in ${\cJ}_{wide}$ and the set of
feasible solutions $S$ with the property that each
interval in $S$ is assigned at least $\eps W$ colors. Note that for each such
feasible solution, at any time $t > 0$, the number
of active intervals is bounded by $1/\eps$, which is a constant. Hence, we can use dynamic programming
to find a solution of maximum profit among all these feasible
solutions.\footnote{This is similar to the dynamic programming algorithm for bounded load instances given in \cite{SWZ13}.
Thus, we omit the details.}
By our assumption, the
value of this solution is at least
$\beta {OPT}_{\fsapss}({\cJ})$.
\hfill \eod
\end{proof}

Next, we consider the complement case.
In Figure \ref{fig:LP_fba} we give the linear program used by
Algorithm ${\acn}$. 
For each $J_i \in {\cJ}_{narrow}$ the linear program has a variable
$x_i$ indicating the number of colors assigned to $J_i$.
For each $J_i \in {\cJ}_{wide}$, the linear program has
two variables: $x_i$ and $y_i$, where $x_i + y_i$ is the number of colors assigned to $J_i$;
$y_i$ gives the number of assigned colors ``over" the first
$\lceil \eps W \rceil$ colors.

Constraint (\ref{const:bound_wide_profit}) bounds the total profit from the extra allocation for each
interval $J_i\in {\cJ}_{wide}$.
The constraints (\ref{const:color_bound}) bound the total number of
colors used at any time $t>0$ by $(1-\eps)W$.
We note that the number of constraints in (\ref{const:color_bound}) is polynomial in $|{\cJ}|$,
since we only need to consider the ``interesting" points of time $t$,
when $t=s_i$ 
for some interval $J_i \in {\cJ}$, i.e., we have at most $n$ constraints.

\label{app:LP_fba}
\begin{figure}[h]
\begin{center}
\begin{minipage}{0.99\textwidth}
\begin{eqnarray}
({LP}_{\fbass}):~~ & \mbox{max} &
\displaystyle{\sum_{J_i \in {\cJ}_{narrow}} p_i x_i + \sum_{J_i \in {\cJ}_{wide}} p_i (x_i + y_i) }
\nonumber
\\
& \mbox{~~~s.t.~~}  & \displaystyle{\sum_{J_i \in {\cJ}_{wide}}
p_i y_i \leq \beta (1-\eps) OPT_{\fsapss}({\cJ})}
\label{const:bound_wide_profit}
\\
& {~~} &
\displaystyle{\sum_{\{ J_i \in {\cJ}: t \in J_i \}} x_i +
              \sum_{\{ J_i \in {\cJ}_{wide}: t \in J_i \}} y_i \leq (1-\eps) W~~}
\hbox{ } \forall  t > 0 ~~
\label{const:color_bound}
\\
& {~~} &
0 \le x_i \le \min\{ r_{max}(i), \lceil \eps W \rceil \} \hbox{~~~~~~~~~~~~~~for } ~ 1\leq i\leq n
\nonumber
\\
& {~~} &
0 \le y_i \le r_{max}(i)- \lceil \eps W \rceil \hbox{~~~~~~~~~~~~~~~~~~for } ~ J_i \in {\cJ}_{wide}
\nonumber
\end{eqnarray}
\end{minipage}
\end{center}
\caption{The linear program $LP_{\fbass}$}.\label{fig:LP_fba}
\end{figure}

\begin{mylemma}
\label{lemma:LP_BP_integral_solution}
For any $\eps > 0$, there is an integral solution of ${LP}_{\fbass}$ of
total profit at least $(1-2\eps){OPT}_{\fsapss}({\cJ})$.
\end{mylemma}

\begin{proof}
Let ${\bar x}^*,{\bar y}^*$ be an optimal (fractional) solution of
${LP}_{\fbass}$. Note that, by possibly move value from $y^*_i$ to
$x^*_i$, we can always find such a solution in which
for any $J_i \in {\cJ}_{wide}$,
if $y^*_i >0$ then $x^*_i = \lceil \eps W \rceil$.

Let $S_w$ be the set of intervals in ${\cJ}_{wide}$ for which $y_i>0$.
Consider now the solution ${\bar x}^*,{\bar z}$,
where $z_i = \lfloor y^*_i \rfloor$, for all $J_i \in S_w$.
We bound the loss due to the rounding down.
\[
\begin{array}{ll}
\sum_{J_i\in S_w} p_i (x_i +\lfloor y_i \rfloor) & \geq
\sum_{J_i\in S_w} p_i (x_i + y_i -1) \\
& \geq
\sum_{J_i\in S_w} p_i (x_i + y_i) (1- \frac{1}{x_i}) \\
& \geq
\sum_{J_i\in S_w} p_i (x_i + y_i) (1- \frac{1}{\lceil \eps W \rceil})
\end{array}
\]
The last inequality follows from the fact that $x_i = \lceil \eps W \rceil$, for any $J_i \in S_w$. Thus, for
$W \geq 1/\eps^2$, we have that the total profit after rounding the $y_i$ values is at least
$(1- 2\eps) {OPT}_{\fbass}({\cJ})$.

Now, consider the linear program $LP_{round}$,
in which we fix
$b_i= \lfloor y^*_i \rfloor$, for $J_i \in S_w$.
\begin{figure}[h]
\begin{center}
\begin{minipage}{0.95\textwidth}
\begin{eqnarray}
({LP}_{round}):~~ & \mbox{max} & \displaystyle{\sum_{i=1}^n p_i x_i}
\nonumber
\\
& \mbox{~~s.t.~~} &
\displaystyle{\sum_{\{ J_i \in {\cJ}: t \in J_i \}} x_i +
              \sum_{\{ J_i \in S_w: t \in J_i \}} b_i \leq (1-\eps) W}
\hbox{ ~} \forall  t > 0
\nonumber
\\
& {~~} &
0 \le x_i \le \min\{ r_{max}(i), \lceil \eps W \rceil \} \hbox{~~~~~~~~~~~~~~for } ~ 1\leq i\leq n
\nonumber
\end{eqnarray}
\end{minipage}
\end{center}
\end{figure}
We have that ${\bar x}^*$ is a feasible (fractional) solution of
${LP}_{round}$ with profit at least
$(1-2\eps) {OPT}_{\fbass}({\cJ}) - \sum_{J_i \in S_w} p_i b_i$.
The rows of the coefficient matrix of ${LP}_{round}$ can be permuted
so that the time points associated with the rows form an increasing
sequence. In the permuted matrix each column has consecutive 1's, thus this matrix
is {\em totally unimodular (TU)}. It follows that
we can find in polynomial time an integral solution, ${\bar x}_I^*$, of the same total profit. Thus,
the integral solution  ${\bar x}_I^*, {\bar z}$ satisfies the lemma.
\hfill \eod
\end{proof}

\begin{mylemma}
\label{lemma:final_coloring_in_acn}
Let $P(c'(S'_w)), P(c'(S'_n))$ be the total profit from the
intervals in the sets $S'_w$ and $S'_n$ in the
circular coloring $c'$ generated in Step \ref{alg:circular_color}. of ${\acn}$. Then,
$(i)$ $P(c''(S'_w)) \geq \frac{3}{4} P(c'(S'_w))$, and $(ii)$ $P(c''(S'_n)) = P(c'(S'_n))$.
\end{mylemma}

\begin{proof}
Property $(i)$ follows from a result of \cite{SWZ13}, which implies that
the non-circular contiguous coloring $c''$, obtained for $S'_w$ in Steps
\ref{alg:wide_intervals_in_lgood}$-$\ref{alg:cut_wide_to_noncircular}
 of ${\acn}$, has a total profit at least $\frac{3}{4} P(c'(S'_w))$. As shown in \cite{SWZ13},
this algorithm can be implemented in polynomial time.

Also, $(ii)$ holds since (in Steps
\ref{alg:narrow_intervals_in_lgood}$-$\ref{alg:recolor_narrow_to_noncircular})
Algorithm ${\acn}$
assigns $|c'(J_i)|$ colors to any $J_i \in S'_n$, i.e., in the resulting non-circular contiguous coloring, $c''$, we
have $|c''(J_i)|=|c'(J_i)|$, $\forall~ J_i \in S'_n$. We note that $c''$
is a valid coloring, since at most one interval (in particular, an interval $J_i \in S'_n$) can be assigned the color
 $\ell_{\mathrm good}$ at any time $t >0$. It follows, that the intervals $J_i$ in $S'_n$ that contain in $c'(J_i)$
$\ell_{\mathrm good}$ form an independent set. Since, for any $J_i \in S'_n$, $r_{max}(i) \leq \lceil \eps W \rceil$, we can assign
to all of these intervals the same set of new colors in the range $\{ \lfloor (1-\eps)W \rfloor +1, \ldots , W \}$.
\hfill \eod
\end{proof}

\noindent
{\bf Proof of Theorem \ref{thm:proper_fsap_ratio}:}
Let $0 < \beta <1$ be the parameter used by Proper\_{\fsap}.
For a correct guess of ${OPT}_{\fsapss}({\cJ})$ and a fixed value of
$\eps >0$, let $\hat \eps = \eps/3$.
\begin{enumerate}
\item[(i)]
If there is an optimal solution in which
the profit from the intervals in ${\cJ}_{wide}$ that are assigned at least
$\eps W$ colors is at least $\beta {OPT}_{\fsapss}({\cJ})$
then, by Lemma \ref{lemma:profitable_wide}, Proper\_{\fsap} finds in Step
\ref{alg:find_wide_solution}
a solution of profit at least $\beta OPT_{\fsapss}({\cJ})$ in polynomial time.
\item[(ii)]
Otherwise, consider the coloring $c''$ output by ${\acn}$. Then, using Lemma \ref{lemma:final_coloring_in_acn},
we have that
\[
\begin{array}{lll}
\displaystyle{\frac{P(c'({\cJ}))}{P(c''({\cJ}))}} & = \displaystyle{\frac{P(c'(S'_w)) +
 P(c'(S'_n))}
{\frac{3P(c'(S'_w))}{4} + P(c'(S'_n))}}
& \leq \displaystyle {\max_{0 < x \leq \beta} \frac{1}{\frac{3x}{4} + 1-x} = \frac{4}{4 -\beta} }
\end{array}
\]
Now, applying Lemma \ref{lemma:LP_BP_integral_solution}, and the rounding in Step \ref{alg:round_wide_color_assign}
 of ${\acn}$ with $\hat \eps$,
we have that
$ P(c'({\cJ})) \geq (1-3 {\hat \eps}) {OPT}_{\fsapss}({\cJ})= (1-\eps){OPT}_{\fsapss}({\cJ})$.
\end{enumerate}
The theorem follows by taking $\beta= \frac{5}{4}$.
\hfill \eod


\section{Uniform {\fsap} Instances}
\label{sec:fsap_ptas}

We now consider {\em uniform} instances of {\fsap}, in which
$r_{max}(i)=\MaxW$, for some $1 \leq \MaxW \leq W$, and $p_i=1$, for all $1 \leq i \leq n$.
Let $k = \lceil W/\MaxW \rceil$.

We first prove that if $k=W/\MaxW$
(i.e., $W$ is an integral multiple of $\MaxW$) then {\fsap} can be solved optimally by finding a
maximum $k$-colorable subgraph in $G$, and assigning to each interval
in this  subgraph $\MaxW$
contiguous colors.
In contrast to this case if $k > W/\MaxW$, then
we show that {\fsap} is NP-Hard, and give a polynomial approximation
scheme to solve such instances.

\begin{mylemma}
\label{thm:uniform_fsap_kint}
An {\fsap} instance where for all $1 \leq i \leq n$,
$r_{max}(i)=\MaxW$, and $W$ is multiple of $\MaxW$ can be solved
exactly in polynomial time.
\end{mylemma}

\begin{proof}
Consider an
{\fbap} instance where for all $1 \leq i \leq n$,
$r_{max}(i)=\MaxW$, and $W$ is multiple of $\MaxW$. We claim that
there exists an optimal solution to this {\fbap} instance in which each
job gets either 0 or $\MaxW$ colors. To see that consider the behavior
of the algorithm Paging\_{\fba} on such an instance, i.e., in which
$r_{min}(i)=0$ and $r_{max}(i)=\MaxW$. Note that when we
start the algorithm $|\AVAIL|$ is a multiple of $\MaxW$ and thus when
the first (left) endpoint $s_1$ is considered the job $J_1$ is
allocated $\MaxW$ colors and  $|\AVAIL|$ remains a multiple of
$\MaxW$. Assume inductively that as the endpoints are scanned all the
jobs that were allocated colors up to this endpoint were allocated
$\MaxW$ colors and that $|\AVAIL|$ is a multiple of $\MaxW$.
If the current scanned endpoint is a right endpoint $e_i$ then by the
induction hypothesis either 0 or $\MaxW$ colors are added to $\AVAIL$
and nothing else is changed. Suppose that the scanned endpoint is a
left endpoint $s_i$. If $|\AVAIL|>0$ then $\MaxW$ colors from $\AVAIL$
are allocated to $J_i$. Otherwise, $J_i$ may be allocated colors that
were previously allocated to another job $J_\ell$, for $\ell<i$.
However, in this case because $r_{min}(\ell)=0$ all the $\MaxW$ colors
allocated to $J_\ell$ will be moved to $J_i$, and the hypothesis still
holds.

The optimal way to assign either 0 or $\MaxW$ colors to each job is
given by computing the maximum $k=(W/\MaxW)$-colorable subgraph of $G$ and
assigning $\MaxW$ colors to each interval in the maximum $k$-colorable
subgraph. Since the assignment of these colors can be done
contiguously it follows that this is also the optimal solution to the
respective {\fsap} instance.
\hfill \eod
\end{proof}

In the Appendix, we give a proof of hardness in case $k>W/\MaxW$.
\begin{theorem}
\label{thm:hard_uniform_fsap}
{\fsap} is strongly NP-hard even if for all $1 \leq i \leq n$
$r_{max}(i)=3$, and $W$ is not divisible by 3.
\end{theorem}
\subsection{An Approximation Scheme}
\label{sec:approx_scheme}
We now describe a PTAS for uniform istances of {\fsap}.
Fix $\eps >0$, and let ${\cJ}$ be a uniform input for {\fsap}.
Denote by $OPT_{\fsapss}(\cJ)$ the value of an
optimal solution for instance $\cJ$ of {\fsap}.

The scheme handles separately  two subclasses of
inputs, depending
on the value of $\MaxW$. First, we consider the case where
$\MaxW$ is {\em large} relative to $W$, or more precisely
$k= \lceil W/\MaxW \rceil \leq 1/\eps$.
We prove in the next lemma that in this case if we partition the
colors into a constant number of
contiguous strips and limit our solution to always assigning
all the colors in the same strip to the same job, we can find a solution
that is at least $(1-\eps)OPT_{\fsapss}({\cJ})$.
The size of each strip (except possibly the last one) is $\lfloor \eps
\MaxW/4 \rfloor \geq 1$.
Since the number of strips is $O(1/\eps^2)$,
we can find an approximation in this case using
dynamic programming as shown in \cite{SWZ13}.

\begin{mylemma}
\label{lemma:MaxLarge_structure}
Any optimal solution for a uniform input ${\cJ}$ for {\fsap} where $k=
\lceil W/\MaxW \rceil \leq 1/\eps$
can be converted to a solution where each interval
is assigned a number of colors that is an integral number of strips, and the total profit is at least
$(1-\eps) OPT_{\fsapss}({\cJ})$.
\end{mylemma}

\begin{proof}
W.l.o.g., we may assume that $\lfloor \eps \MaxW/4 \rfloor \geq 1$, else we have that $W$ is a constant,
and we can solve {\fsap} optimally in polynomial time (see \cite{SWZ13}).
Given an optimal solution for a uniform input ${\cJ}$, let $S$ be the subset of intervals $J_i$ for which
$|c(J_i)| >0$, and let $G_S$ be the support graph for $S$ (i.e., the
subgraph of the original interval graph induced by the intervals in
$S$). Using the above partition of the colors to strips, we have
$1 \leq N \leq \lceil \frac{4W}{\eps \MaxW} \rceil$ strips. Denote by $C_j$ the subset of colors of strip $j$.
We obtain the {\em strip structure} for the solution as follows.
Let $S_j \subseteq S$ be the subset of intervals with colors in strip $j$, i.e.,
$S_j = \{ J_i| c(J_i) \cap C_j \neq \emptyset\}$. Initialize for all $J_i \in S$, $c'(J_i)=\emptyset$.
$(i)$ For all $1 \leq j \leq N$, find in $G_S$ a maximum independent set, ${\cI}_j$ of
intervals $J_i \in S_j$. For any $J_i \in  {\cI}_j$, assign to $J_i$ all colors in $C_j$, i.e.,
$c'(J_i)=c'(J_i) \cup C_j$.
$(ii)$ For any $J_i \in S$, if  $|c'(J_i)| > \MaxW$ then omit from the coloring of $J_i$
a consecutive subset of strips, starting from the highest $1 \leq j \leq N$, such that $C_j \subseteq c'(J_i)$,
until for the first time $|c'(J_i)| \leq \MaxW$.

We show below that the above {\em strip coloring} for $S$ is feasible, and that the total profit from the
strip coloring is at least $(1-\eps){OPT}_{\fsapss}(J)$.
To show feasibility, note that if $J_i \in S_j$ and $J_i \in S_{j+2}$
then because the coloring is contiguous $J_i$ is allocated all the
colors in $S_{j+1}$,
and thus any maximum independent set ${\cI}_{j+1}$ will contain $J_i$ since it
has no conflicts with other jobs in $S_{j+1}$. It follows that if a job $J_i$ is
allocated colors in more than two consecutive strips, i.e.,
$J_i \in S_j \cap S_{j+1} \cap \cdots \cap S_{j+\ell}$, for $\ell > 1$,
then $J_i \in {\cI}_{j+1} \cap \cdots \cap {\cI}_{j+\ell-1}$.
Thus, each interval $J_i \in S$ will be assigned in step $(i)$ a consecutive set of strips. Hence $c'$ is a contiguous coloring.
In addition, after step $(ii)$, for all $J_i \in S$ we have that
$|c'(J_i)| \leq \MaxW$.

Now, consider the profit of the strip coloring. We note that after step $(i)$, the total profit of $c'$ is at least
${OPT}_{\fsapss}(J)$. This is because for each strip $j$,
$|C_j|\cdot |{\cI}_j|$ is an upper bound on the profit that can be
obtained from this strip. We show that the harm of reducing the number of colors in step $(ii)$ is small.
We distinguish between two type of intervals in $S$.
\begin{enumerate}
\item[(a)]
Intervals $J_i$ for which $|c(J_i)| \geq (1-\eps/2)\MaxW$.
Since coloring $c$ is valid it follows that
$|c'(J_i)|$ is reduced in step $(ii)$  only if before this step
$|c'(J_i)| > |c(J_i)|$.
\comment{
Suppose that $\{ C_\ell, \cup \cdot \cdot \cdot \cup C_{\ell+r-1} \} \subseteq c(J_i)$, for
some $\ell$, and $r \geq 1$, i.e., $J_i$ was assigned all the colors in $r$ consecutive strips, and possibly some colors
in $C_{\ell -1}$, and some colors in  $C_{\ell +r}$. Thus, after step $(ii)$, we may have that
$c'(J_i) \cap \{ C_{\ell+r-1} \cup C_{\ell +r} \} = \emptyset$.
}
Consider all the strips that contain colors assigned to $J_i$ in the
original coloring $c$. Note that in all such strips except at most two
no colors are assigned to any interval that intersects $J_i$.
Thus $|c'(J_i)|$ is reduced in step $(ii)$  by at most
$2 \lfloor \frac{\eps \MaxW}{4} \rfloor \leq (\eps/2) \MaxW$.
Since $|c(J_i)| \geq (1- \eps/2)\MaxW$, we have that after step
$(ii)$, $|c'(J_i)| \geq |c(J_i)|- (\eps/2) \MaxW \geq (1-\eps) \MaxW$.
\item[(b)]
For any interval $J_i$ for which $|c(J_i)| <
(1-\eps/2)\MaxW$, since after
step $(i)$  $|c'(J_i)| \leq |c(J_i)|+2 \lfloor \frac{\eps \MaxW}{4} \rfloor$, we have that
$|c'(J_i)| \leq \MaxW$. Thus, no colors are omitted from $J_i$ in step $(ii)$.
\end{enumerate}
From $(a)$ and $(b)$, we have that the total profit from the strip coloring satisfies
${OPT'}_{\fsapss}({\cJ}) \geq (1-\eps) {OPT}_{\fsapss}({\cJ})$.
\hfill \eod
\end{proof}

Now, consider the case where $\MaxW$ is {\em small}, i.e.,
$k= \lceil W/\MaxW \rceil > 1/\eps$. In this case we consider just
$(k-1)\MaxW$ consecutive colors and ignore the remainder up to
$\eps W$ colors. Recall that when the number of colors is a multiple
of $\MaxW$ we can find an optimal solution.
Let $OPT_{\fsapss(W)}({\cJ})$ denote the
value of an optimal solution for instance $\cJ$ of {\fsap} with $W$
colors, and recall that $k =\lceil W/\MaxW \rceil$.
Since $(k-1)\MaxW < W < k\MaxW$,
$OPT_{\fsapss((k-1)\MaxW)}({\cJ})<OPT_{\fsapss(W)}({\cJ})<OPT_{\fsapss(k\MaxW)}({\cJ})$.
The value of $OPT_{\fsapss((k-1)\MaxW)}$  is $\MaxW$ times the size of
the max $(k-1)$-colorable subgraph of $G$, and
the value of $OPT_{\fsapss(k\MaxW)}$  is $\MaxW$ times the size of
the max $k$-colorable subgraph of $G$. Clearly, the ratio of the sizes
of these subgraphs and thus the ratio of the two optimal values is
bounded by $1 - 1/k>1-\eps$. It follows that
$OPT_{\fsapss((k-1)\MaxW)}({\cJ}) \ge
(1-\eps)OPT_{\fsapss(k\MaxW)}({\cJ}) > (1-\eps)OPT_{\fsapss(W)}({\cJ})$.
Combining the results, we have
\begin{theorem}
The above algorithm is a PTAS for uniform {\fsap} instances.
\end{theorem}

\myparagraph{Acknowledgments}
We thank Magn\'{u}s \ Halld\'{o}rsson and Viswanath Nagarajan for valuable discussions.
\bibliographystyle{abbrv}
\bibliography{fsap}

\normalsize
\def\baselinestretch{1}
\appendix

\section{Algorithm Paging\_{\fba}}

Algorithm~\ref{alg:opt_fba} gives the pseudocode of Paging\_{\fba}.
\begin{algorithm}[!h]
\caption{Paging\_{\fba}$({\cJ}, {\bar r_{min}}, {\bar r_{max}}, W)$ \label{alg:opt_fba}}
  \begin{algorithmic}[1]
    \STATE $\SOL=0$.
    \STATE $\AVAIL=[1..W]$.
    \STATE W.l.o.g. assume that all endpoints are distinct.
    Sort the left endpoints of the intervals in ${\cJ}$ in
    non-decreasing order. Let $L$ denote the sorted list.
    \WHILE{not end-of-list}
        \STATE Consider the next endpoint $L$.
        \IF{the endpoint is $s_i$, the left endpoint of $J_i$,}
            \STATE $\SOL += r_{min}(i)$.
            \STATE $A_i=\emptyset$.
            \IF{$|\AVAIL|>r_{min}(i)$}
                \STATE Move $r_{min}(i)$ colors from $\AVAIL$ to $A_i$.
            \ELSE
                \STATE $A_i = \AVAIL$
                \STATE $\AVAIL=\emptyset$.
                \WHILE{$|A_i|<r_{min}(i)$}
                    \STATE Among all intervals such that $s_k<s_i$ and $|A_k|>r_{min}(k)$,
                    let $J_k$ be the interval with maximum right endpoint $e_k$.
                    \STATE Move $\min\{|A_k|-r_{min}(k),r_{min}(i)-|A_i|\}$ colors from $A_k$ to $A_i$.
                \ENDWHILE
            \ENDIF
            \STATE Move $\min\{|\AVAIL|,r_{max}(i)-r_{min}(i)\}$ colors from $\AVAIL$ to $A_i$.
            \STATE $\SOL += \min\{|\AVAIL|,r_{max}(i)-r_{min}(i)\}$.
            \STATE Let ${\cJ}_i = \{J_k \in \cJ | s_k < s_i \wedge e_k > e_i \wedge |A_k|>r_{min}(k)\}$.
            \WHILE {$|A_i|<r_{max}(i) \wedge |\cJ_i|>0$}
                \STATE Let $J_k$ be the interval with the maximum right endpoint in $\cJ_i$.
                \STATE Move $\min\{|A_k|-r_{min}(k),r_{max}(i)-|A_i|\}$ colors from $A_k$ to $A_i$
            \ENDWHILE
            \ELSE
            \STATE  \COMMENT{The endpoint is $e_i$, the right endpoint of $J_i$}
            \STATE $\AVAIL = \AVAIL \cup A_i$
        \ENDIF
    \ENDWHILE
    \STATE Return $\SOL$ and the sets $A_1, \ldots A_n$ of colors assigned to the jobs in ${\cJ}$.
   \end{algorithmic}
\end{algorithm}

\section{Hardness of {\fsap}  for Uniform Instances}

\noindent
{\bf Proof of Theorem \ref{thm:hard_uniform_fsap}: }
The reduction is from of the 3-Exact-Cover (3XC) problem defined as
follows. Given a universal set $U=\{e_1,\ldots,e_{3n}\}$ and a
collection
$S=\{S_1,\ldots,S_m\}$ of 3 element subsets of $U$, is there a
sub-collection $S' \subseteq S$ such that each element of $U$ occurs
in exactly one member of $S'$. Note that $|S'|=n$, if such exists.
Recall that Karp showed in his seminal paper~\cite{K72} that 3XC is
NP-Complete.

To simplify the presentation, we first assume that the intervals have
different profits per allocated unit.
For the reduction, we use several sets of intervals. One such set is shown
in Fig.~\ref{fig:FSASetJobs}. It consists of 16 intervals whose lengths and
relative positions are given in the figure. Assume that the profit of
each of the
intervals $9,\ldots,16$ is higher than the sum of the profits
of intervals $1,\ldots,8$, and that
the profit of each of the intervals $2$ to $7$ is higher than the
profit of intervals $1$ and $8$.

\begin{figure}[h]
\begin{center}
\includegraphics[scale=.75]{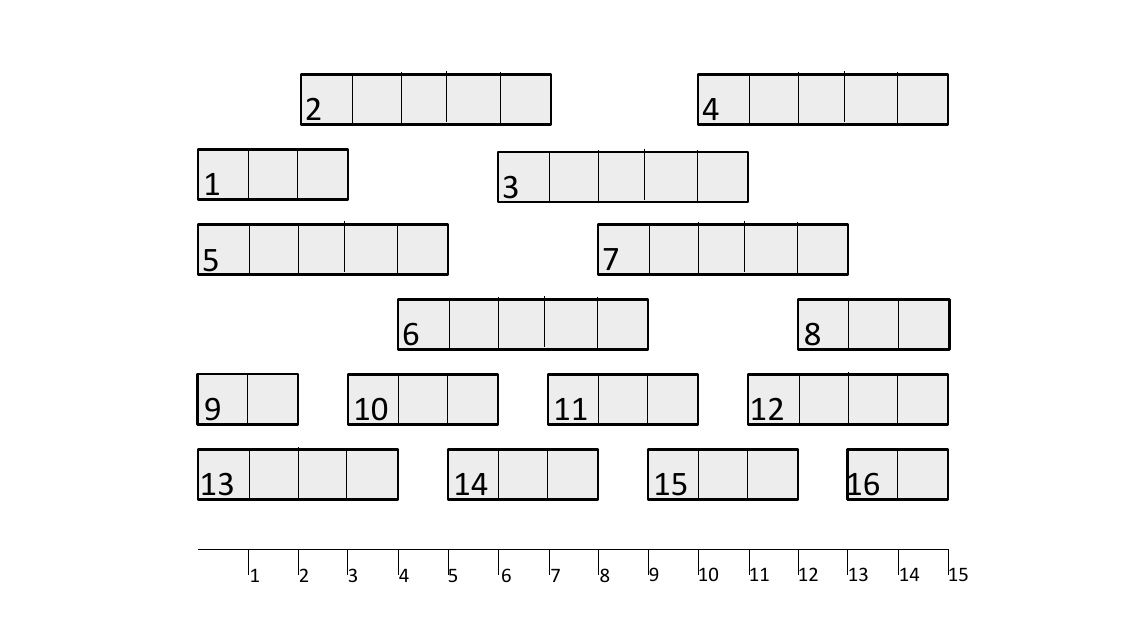}
\caption{The ``two-choice" set of intervals}\label{fig:FSASetJobs}
\end{center}
\end{figure}

Suppose that we are given two ``banks" of contiguous colors to allocate to this
set of intervals: one bank consists of four contiguous colors and one
consists of three contiguous colors. Given that $r_{max}(i)=3$, our
first priority is to allocate three colors to each interval in
$[9..16]$.
Assuming that three colors are indeed allocated to each interval in
$[9..16]$,
note that any other interval can be allocated at most one color. This
is since
any other interval intersects two intervals from $[9..16]$ at
a point.

We say that two intervals from $[1..8]$ {\em conflict}
if both cannot be allocated colors simultaneously. Note that two
such intervals conflict if both intersect two intervals from
$[9..16]$ at the same time point because only 7 colors are available. Define the conflict graph to
be a graph over the vertices $[1..8]$, where two vertices
are connected if the respective intervals conflict. It is easy to see
that the conflict graph is the path \mbox{1 -- 5 -- 2 -- 6 -- 3 -- 7 -- 4 -- 8}.
Since the profit of intervals $[2..7]$ is higher than the profit of
intervals $1$ and $8$, the best strategy is to allocate one color to
three of the intervals in $[2..7]$ and one color to either interval $1$
or $8$. The only two possibilities for allocating the two banks of colors
in such a way are illustrated in Fig.~\ref{fig:FSASetJobsAOpt12}.
Because of this property, we call this set of intervals a ``two-choice"
gadget.
Note that, in the first option, the bank of 3 colors is unassigned
at times: $5$, $9$ and $13$, while in the second option, it is
unassigned at times $3$, $7$ and $11$.

\begin{figure}[h]
\begin{center}
\includegraphics[scale=.675]{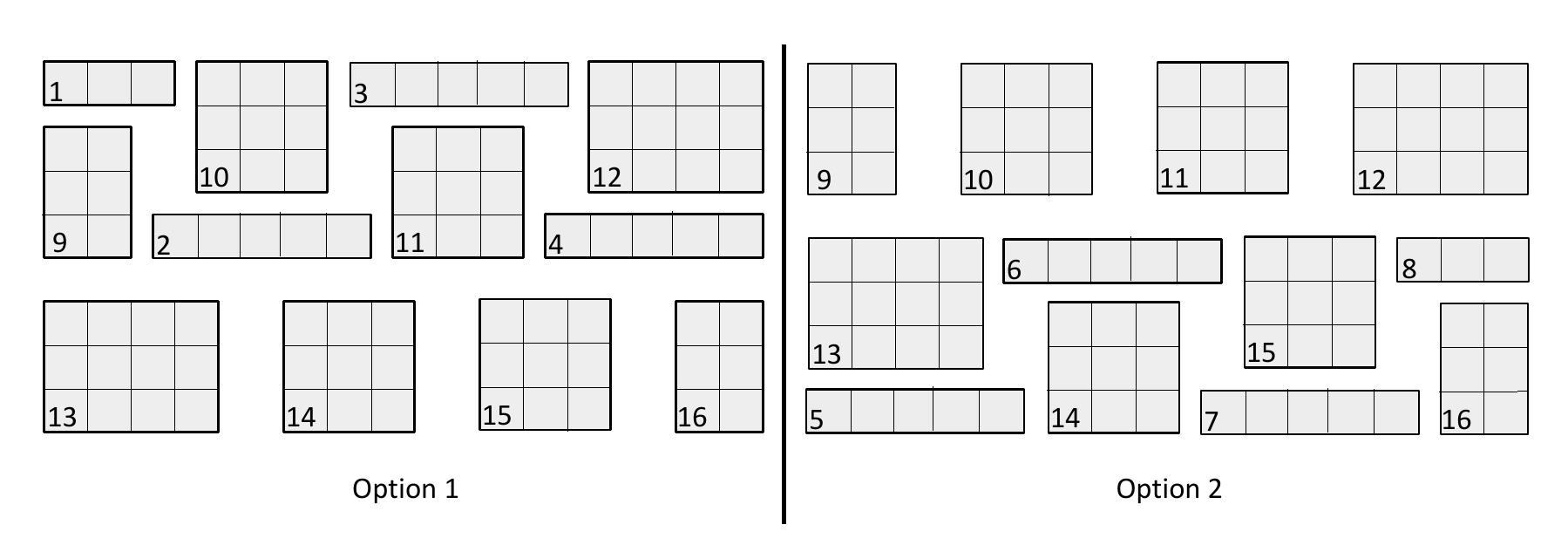}
\caption{The set of intervals}\label{fig:FSASetJobsAOpt12}
\end{center}
\end{figure}

We need to define also a pair of intervals called an ``overlapping"
pair of intervals, illustrated in  Fig.~\ref{fig:FSAOverlapJobs}.
Note that to be able to allocate 3 colors to both intervals, we need one
bank of 3 contiguous colors at time interval $[t_1,t_2+1)$ and another
at time interval $[t_2,t_3)$; that is, we need both banks at time
interval $[t_2,t_2+1)$.

\begin{figure}[h]
\begin{center}
\includegraphics[scale=.75]{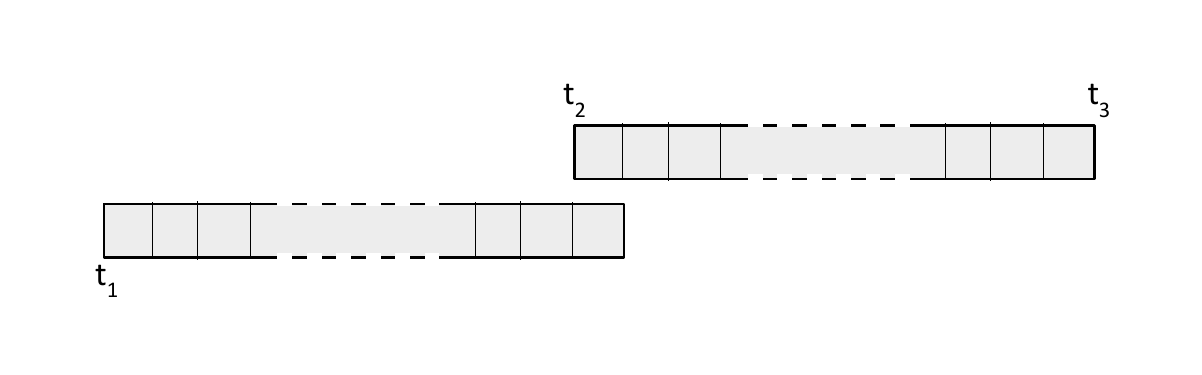}
\caption{Overlapping intervals}\label{fig:FSAOverlapJobs}
\end{center}
\end{figure}

We now describe the reduction from the 3XC problem. For each set
$S_i \in S$, associate a ``two-choice"
gadget. The ``two choices" correspond
to the decision whether to include $S_i$ in the cover or not. For each
element $e \in U$, and for each set $S_i$ such that $e \in S_i$, we
associate a pair of overlapping intervals, where the overlap is in one
of the times in which the ``two choice" gadget corresponding to set
$S_i$ has an unassigned bank of 3 colors. In addition, we define some
extra intervals as detailed below.

Given a 3XC problem instance, set $W=9m+7$ and $r_{max}=3$ for all
intervals. Let $P=8n+45m$. In the reduction, we have 5 groups of intervals defined as
follows.
\begin{enumerate}
\item Left border: $3m+3$ intervals $L_1,\ldots , L_{3m+3}$ each of
    profit $P^2$.
    For $i \in [1..3n]$, $L_i= [0, i)$,
    for $i \in [3n+1..3m]$, $L_i = [0,4n)$, and
    for $i \in [3m+1..3m+3]$, $L_i = [0,4n+15m)$.
\item Right border: $3m+3$ intervals $R_1,\ldots R_{3m+3}$ each of
    profit $P^2$.
    For $i \in [1..3n]$, $R_i= [5n+45m+i, 8n+45m+1)$,
    for $i \in [3n+1..3m]$, $R_i = [4n+45m, 8n+45m+1)$, and
    for $i \in [3m+1..3m+3]$, $R_i = [4n+30m, 8n+45m+1)$.
\item ``two choice" gadgets: $m$ copies of the ``two choice" gadget,
    one for each set $S_i \in S$. The gadget associated with set $S_i$
    starts at time $4n+15m+15(i-1)$ and its length is 15 time units. The
    profit of intervals $1$ and $8$ in each copy is $1$, the profit of
    intervals $2$ to $7$ is $2$, and the profit of intervals $9$ to
    $16$ is $P^2$.
\item element overlapping pairs: $3m$ overlapping pairs of intervals,
    one per occurrence of an element in a set.
    For each $S_i \in S$, let $S_i = \{e_a,e_b,e_c\}$, where
    $\{a,b,c\} \subseteq [1..3n]$. The respective three overlapping pairs
    are (1) $[a,4n+15m+15(i-1)+3)$ and $[4n+15m+15(i-1)+2,5n+45m+a)$,
    (2) $[b,4n+15m+15(i-1)+7)$ and $[4n+15m+15(i-1)+6,5n+45m+b)$, and
    (3) $[c,4n+15m+15(i-1)+11)$ and $[4n+15m+15(i-1)+10,5n+45m+c)$.
    The profit of each such interval is its length. Note that the
    profit of any pair of overlapping intervals is $5n+45m+1$.
\item ``filler" overlapping pairs: $3m$ overlapping pairs, three per
    set. For each $S_i \in S$, the respective three overlapping pairs
    are (1) $[4n,4n+15m+15(i-1)+5)$ and $[4n+15m+15(i-1)+4,4n+45m)$,
    (2) $[4n,4n+15m+15(i-1)+9)$ and $[4n+15m+15(i-1)+8,4n+45m)$, and
    (3) $[4n,4n+15m+15(i-1)+13)$ and $[4n+15m+15(i-1)+12,4n+45m)$.
    The profit of each such interval is its length. Note that the
    profit of any pair of overlapping intervals is $45m+1$.
\end{enumerate}

\begin{mylemma}
    \label{lem:3XC_red}
    The 3XC instance has an exact cover if and only if the associated
    {\fsap} instance  has profit
    $(18m+14+24m)P^2+7m+9n(5n+45m+1)+(9m-9n)(45m+1)=
    (42m+14)P^2+405m^2+45n^2+16m$.
\end{mylemma}
\begin{proof}
    Assume that the 3XC instance has an exact cover. We show an
    assignment of the intervals in the {\fsap} instance that achieves
    the desired profit.
    First, assign color $1$ to $L_{3m+2}$ and $R_{3m+2}$ and colors
    $2,3,4$ to $L_{3m+3}$ and $R_{3m+3}$. Also assign colors $5,6,7$
    to $L_{3m+1}$ and colors $9m+5,9m+6,9m+7$ to $R_{3m+1}$.
    Now, consider $S_i$, for $i \in [1..m]$. Assume that $0\le k <i$
    sets $S_j$, for $j<i$ are in the cover. (See also
    Fig.~\ref{fig:FSASetJobsAOpt12a}.)

\begin{figure}[h]
\begin{center}
\includegraphics[scale=.615]{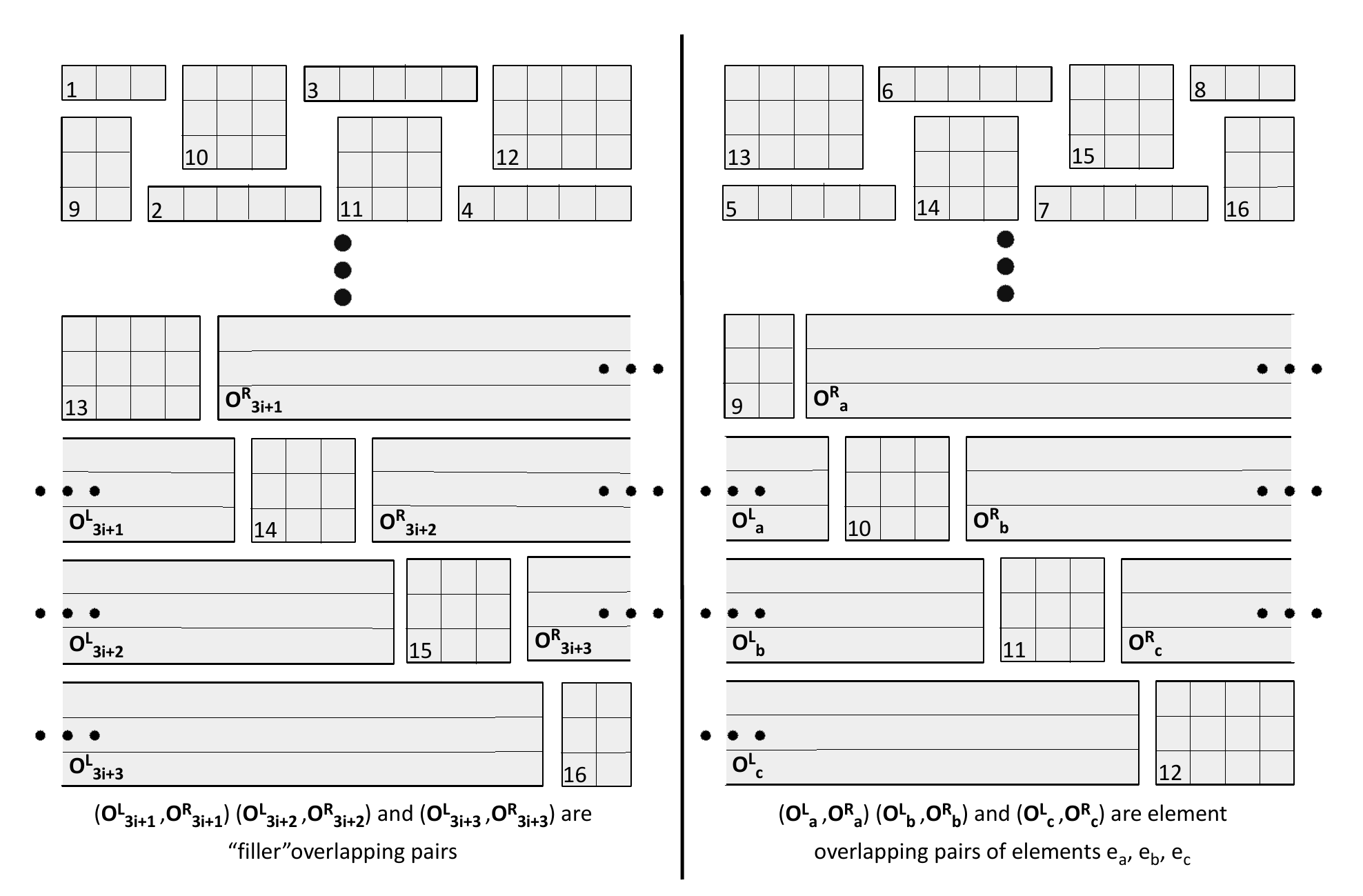}
\caption{Assigning intervals related to $S_i$}\label{fig:FSASetJobsAOpt12a}
\end{center}
\end{figure}

    If $S_i$ is not in the cover
    choose the first option for the ``two choice" gadget associated
    with $S_i$. Namely, assign colors $[1..4]$ to intervals $[1..4]$
    and $[9..12]$ in the gadget, and for $j \in [13..16]$ assign colors
    $9i+3j-43,9i+3j-42,9i+3j-41$ to interval $j$.
    Also, assign colors $9i-4,9i-3,9i-2$ to $R_{3(n+i-k)-2}$;
    for $j \in \{2,3\}$, assign colors $9i+3j-7,9i+3j-6,9i+3j-5$
    to $L_{3(n+i-k)+j-4}$ and $R_{3(n+i-k)+j-3}$,
    and assign colors $9i+5,9i+6,9i+7$ to $L_{3(n+i-k)}$.
    Finally, assign 3 colors to the 3 ``filler" overlapping pairs
    corresponding to $S_i$ as follows: colors $9i-4,9i-3,9i-2$ to
    interval $[4n+15m+15(i-1)+4,4n+45m)$, colors $9i-1,9i,9i+1$ to
    intervals  $[4n,4n+15m+15(i-1)+5)$  and $[4n+15m+15(i-1)+8,4n+45m)$,
    colors $9i+2,9i+3,9i+4$ to intervals
    $[4n,4n+15m+15(i-1)+9)$  and $[4n+15m+15(i-1)+12,4n+45m)$,
    and  colors
    $9i+5,9i+6,9i+7$ to interval $[4n,4n+15m+15(i-1)+13)$.

    Suppose that $S_i$ is in the cover. Let $S_i = \{e_a,e_b,e_c\}$, where
    $\{a,b,c\} \subseteq [1..3n]$.
    Choose the second option for the ``two choice" gadget associated
    with $S_i$. Namely, assign colors $[1..4]$ to intervals $[5..8]$
    and $[13..16]$ in the gadget, and for $j \in [9..12]$ assign colors
    $9i+3j-31,9i+3j-30,9i+3j-29$ to interval $j$.
    The respective element overlapping pairs and the
    border intervals are assigned colors as follows:
    colors $9i-4,9i-3,9i-2$ to
    $[4n+15m+15(i-1)+2,5n+45m+a)$ and $R_a$,
    colors $9i-1,9i,9i+1$ to $L_a$, $[a,4n+15m+15(i-1)+3)$,
    $[4n+15m+15(i-1)+6,5n+45m+b)$ and $R_b$, colors
    $9i+2,9i+3,9i+4$ to $L_b$, $[b,4n+15m+15(i-1)+7)$,
    $[4n+15m+15(i-1)+10,5n+45m+c)$ and $R_c$,  and colors
    $9i+5,9i+6,9i+7$ to $L_c$ and $[c,4n+15m+15(i-1)+11)$.

Note that the assignment is valid, that is, no two overlapping
intervals are assigned the same color. To calculate the total profit
of the assignment note that $L_{3m+1}, L_{3m+3}, R_{3m+1}$ and
$R_{3m+3}$ are each assigned 3 colors and $L_{3m+2}, R_{3m+2}$ are
assigned a single color. This contributes
$14P^2$ to the profit. Also, intervals $[9..16]$ in all the ``two
choice" gadgets are assigned 3 colors, and either intervals $[1..4]$
or $[5..8]$ are assigned a single color. This contributes $24mP^2+7m$ to
the profit. Since we start from a cover, all the element overlapping pairs
as well as the corresponding left and right border intervals are
assigned 3 colors each.
This contributes $3\cdot 6nP^2+3 \cdot 3n(5n+45m+1)$
to the profit. Since the cover is exact, $3m-3n$ ``filler" overlapping
pairs as well as the corresponding left and right border intervals are
assigned 3 colors each. This contributes $3(6m-6n)P^2+3(3m-3n)(45m+1)$
to the profit. Overall, the profit is
$(42m+14)P^2+405m^2+45n^2+16m$.

We now prove the other direction. Suppose that we find an assignment
of colors to intervals with total profit
$(42m+14)P^2+405m^2+45n^2+16m$.
The only way to get total profit of at least $(42m+14)P^2$ is by
assigning 3
colors to all intervals with $P^2$ profit per allocated unit in the ``two choice"
gadgets, and in addition, by assigning 3 colors to all but one left
(and right) border intervals, and by assigning the 1 remaining color to the
remaining left (and right) border interval.

Consider the (element and ``filler") overlapping pairs. At most $3m$ left (and right)
intervals out of these overlapping pairs can be assigned 3 colors
each, since any additional assignment would conflict with the assignment
of colors to the border intervals.
Out of these $3m$ left and right intervals, at most $3n$ left (and right)
intervals can be element overlapping intervals.

Since all intervals with $P^2$ profit per allocated unit in the ``two choice"
gadgets are assigned 3 colors, there are $3m$ remaining 3 color blocks
throughout the interval $[4n+15m,4n+30m)$ and at most one
more block of $3$
colors is available in each of the $6m$ time units when some of
the intervals with $P^2$ profit per allocated unit in the ``two choice"
gadgets are not assigned any color (see Fig.~\ref{fig:FSASetJobsAOpt12}).
The maximum profit that can be attained by assigning these colors to
the unassigned intervals in the ``two choice"
gadgets is at most $3\cdot (2\cdot6 +2)m= 42m$. 
Thus the only way to achieve the $405m^2$ term in the profit (for
large enough $m$) is by
actually assigning 3 colors to $3m$ left and $3m$
right intervals out of the overlapping pairs.

Consider the $3m$ left intervals of the overlapping pairs that are assigned 3 colors in
increasing length order and the right intervals of the overlapping pairs
that are assigned 3 colors in decreasing length order.
Denote these two sequences by $O^L_1,\dots,O^L_{3m}$ and
$O^R_1,\dots,O^R_{3m}$.
\begin{cl}
\label{claim:unit_overlap}
For $i \in [1..3m]$, if $O^L_i$ and $O^R_i$ overlap, they cannot
overlap by more than one time unit.
\end{cl}

\begin{proof}
Suppose the claim does not hold, and let $i$ be the minimum index for
which  $O^L_i$ and $O^R_i$ overlap by more than one time unit.
However, in
this case $O^L_i$ and $O^R_i$ must overlap in at least one time unit
$t$ when
the intervals with $P^2$ profit per allocated unit in the ``two choice"
gadgets are assigned two blocks of 3 colors. Since $O^R_i$ contains
time $t$, $t$ is contained also in $O^R_1,\ldots,O^R_{i-1}$.
Similarly, since $O^L_i$ contains
time $t$, $t$ is contained also in $O^L_{i+1},\ldots,O^L_{3m}$.
But this implies that $3m+1$ intervals out of the overlapping pairs
and 2 intervals from the ``two choice" gadget are each assigned 3
colors. This is impossible, since there are only $9m+7$ colors.
\hfill \eod
\end{proof}

From the discussion above, it follows that
if $O^L_i$ and $O^R_i$ overlap they must be an
overlapping pair.
The maximum profit that can be attained from the intervals in the ``two choice"
gadgets that do not have $P^2$ profit per allocated unit is $14m$. Thus, to achieve
the additional $405m^2+45n^2$ terms in the profit, we must have that
for all $i \in [1..3m]$, intervals $O^L_i$ and $O^R_i$ are an
overlapping pair, and $3n$ out of these overlapping pairs are
element overlapping pairs.
This will contribute  $405m^2+45n^2+9m$ to the profit. The extra $7m$
profit needs to be attained by assigning colors to the intervals in the ``two choice"
gadgets that do not have $P^2$ profit per allocated unit. It follows
that each ``two
choice" gadget needs to be colored using one of the two options described
above and exactly $n$ of them have to be colored using the second
option. These $n$ gadgets correspond to the exact cover.
\hfill \eod
\end{proof}

Finally, we note how the reduction can be modified to include only
intervals of identical profit per allocated unit.
The idea is to ``slice" each interval in the original reduction to
smaller intervals whose number is the profit per allocated unit of the original interval.
When doing so, we need to ensure that it is not beneficial to
move from a ``slice" of one interval to a ``slice" of another
interval. This is done by assigning different displacements to the slices
in different intervals, so that whenever we attempt to gain from a move
from a ``slice" of one interval to a ``slice" of another
interval, we lose at least one slice due to the different
displacements. Thus, the same set of colors will be used for all slices associated with the original interval.
This completes the proof of the theorem.
\hfill \eod

\section{Algorithm ${\acn}$}

The pseudocode of ${\acn}$ is given in Algorithm \ref{alg:color_narrow}.

\begin{algorithm}[!h]
\caption{${\acn}({\cJ}, {\bar r_{max}}, {\bar p},W, \eps, P))$
 \label{alg:color_narrow}}
  \begin{algorithmic}[1]
   \STATE Solve the linear program ${LP}_{\fbass}$

   \STATE Round the solution to obtain a (non-contiguous) coloring  $c$.\label{alg:round_LP_BP_solution}
   \STATE Find a circular contiguous coloring $c'$ of the same total
   profit as $c$. Such a coloring exists since the input graph is
   proper, and is obtained by scanning the intervals from left to right and
   assigning the colors in a fixed circular order.
\label{alg:circular_color}
   \STATE Let $S'$ be the set of colored intervals in $c'$.
   \STATE Let $S'_w \subseteq S'$ be the subset of intervals $J_i \in S'$ for which $|c'(J_i)| \geq \eps W$,
               and let $S'_n = S' \setminus S'_w$.
\STATE for any $J_i \in S'_w$, round down $|c'(J_i)|$ to the nearest integral multiple of $\eps^2 W$, and
eliminate the corresponding amount of colors in $c'(J_i)$, such the first color assigned to $J_i$ is
$f_i = \lfloor r \cdot \eps^2W \rfloor$, for some integer $r \geq 0$.
\label{alg:round_wide_color_assign}
 \FOR  {$\ell = \lfloor  r \cdot \eps^2 W \rfloor$, $r=0,1 , \ldots,  \lceil \frac{1}{\eps^2} \rceil $}
    \STATE Let $S'_w(\ell) = \{ J_i \in S'_w | \{\ell, \left[ \ell \bmod
        \lfloor (1-\eps)W \rfloor \right]+1 \} \subseteq  c'(J_i) \}$
    \STATE Let $P(S'_w(\ell))= 0$
    \FOR {$J_i  \in S'_w(\ell)$}
        \STATE Suppose that $c'(J_i) = \{ f_i, \left[ f_i \bmod \lfloor (1-\eps)W
         \rfloor \right] +1, \ldots , t_i \}$,\par
            \hskip\algorithmicindent for some $1 \leq f_i, t_i \leq \lfloor (1-\eps)W \rfloor$.
        \STATE Partition the set of $|c'(J_i)|$ colors assigned to
        $J_i$ into two contiguous blocks:\par
        \hskip\algorithmicindent ${Block}_1(i)=\{f_i, \ldots, \ell \}$, and ${Block}_2(i)=\{\ell \bmod \lfloor (1-\eps)W \rfloor +1, \ldots , t_i \}$.
        \STATE Let $P(S'_w(\ell)) += p_i \cdot \min\{|{Block}_1(i)|,|{Block}_2(i)|\}$
    \ENDFOR
\ENDFOR
\STATE Let $\ell_{\mathrm good} = {\mathrm {argmin}}_{1 \leq \ell \leq \lfloor (1-\eps)W \rfloor} P(S'_w(\ell))$.
\FOR {$J_i  \in S'_w(\ell_{\mathrm good})$} \label{alg:wide_intervals_in_lgood}
    \STATE Assign to $J_i$ the larger of ${Block}_1(i)$ and ${Block}_2(i)$. \label{alg:cut_wide_to_noncircular}
\ENDFOR
\STATE Renumber the first $ \lfloor (1-\eps)W \rfloor$ colors starting at
$\ell_{\mathrm good} \bmod \lfloor (1-\eps)W \rfloor +1$.
\STATE Let $S'_n(\ell_{\mathrm good}) = \{ J_i \in S'_n | \{\ell_{\mathrm good},
    \left[ \ell_{\mathrm good} \bmod \lfloor (1-\eps)W \rfloor \right] +1 \} \subseteq  c'(J_i) \}$
\FOR {$J_i \in S'_n(\ell_{\mathrm good})$} \label{alg:narrow_intervals_in_lgood}
    \STATE Assign to $J_i$ $|c'(J_i)|$ contiguous colors in the set
        $\{\lfloor (1-\eps)W \rfloor+1, \ldots ,W \}$. \label{alg:recolor_narrow_to_noncircular}
\ENDFOR
\STATE Return $c''$ the resulting coloring of $S'$.
\end{algorithmic}
\end{algorithm}

\section{A $\frac{2k}{2k-1}$-approximation Algorithm for Uniform {\fsap}}

We describe below Algorithm ${\ca}_{MaxSmall}$ for uniform {\fsap} instances.
It yields solutions that are close to the optimal
as $k=\lceil \frac{W}{\MaxW} \rceil$ gets large. Initially, ${\ca}_{MaxSmall}$ solves optimally {\fbap} on
 the input graph $G$. Let $G'$ be the support graph for this solution. ${\ca}_{MaxSmall}$ proceeds by generating
a  feasible solution for {\fsap} as follows. Consider the set of
intervals $J_i \in G'$ for which $|c(J_i)|=\MaxW$.
Denote the support subgraph of this set of intervals $G_2$.
Observe that the graph $G_2$ is $(k-1)$-colorable, since there is no
clique of size $k$ in $G_2$ (as it would require $k\MaxW > W$ colors).
Consequently, each interval in $G_2$ can be assigned a set of
$\MaxW$ contiguous colors, using a total of $(k-1)\MaxW$ colors. ${\ca}_{MaxSmall}$ then finds a
maximum independent set of intervals in the remaining subgraph $G_1 = G' \setminus G_2$, and
colors contiguously each interval in this set using the remaining $W \bmod \MaxW$ colors (see the pseudocode in
Algorithm \ref{alg:MaxSmall}).

\begin{algorithm}[!htb]
\caption{${\ca}_{MaxSmall}({\cJ}, \MaxW, W)$ \label{alg:MaxSmall}}
  \begin{algorithmic}[1]
 \STATE Find an optimal solution for {\fbap} on $G$, the interval graph
 of ${\cJ}$, using Algorithm Paging\_{\fba}.
\STATE Let $S \subseteq {\cJ}$ be the solution set of intervals, and $G' \subseteq G$ the support graph of $S$.
\STATE Let  $G_2 \subseteq G'$ the subgraph induced by the intervals
$J_i$ for which $|c(J_i)|=\MaxW$. \label{step:def_G2}
\STATE  Let  $G_1 \subseteq G'$ be the subgraph induced by the rest of
the intervals, i.e., intervals
$J_i$ for which $|c(J_i)|<\MaxW$. \label{step:def_G1}
\STATE Scan the intervals $J_i$ in $G_2$ from left to right and color contiguously
each interval with the lowest
available $\MaxW$ colors.\label{step:color_G2}
\STATE Let $r=W \bmod \MaxW$.
\STATE Let $I$ be a maximum independent set in $G_1$.
\STATE Color each interval in $I$ contiguously with $r$ colors.
\STATE Return the coloring of the intervals in $G_1 \cup G_2$.
   \end{algorithmic}
\end{algorithm}

\remove{
\myparagraph{Analysis}
We now analyze Algorithm ${\ca}_{MaxSmall}$.
}

\begin{theorem}
\label{thm:ratio_MaxSmall}
For any uniform instance of {\fsap}, ${\ca}_{MaxSmall}$ yields an
optimal solution for {\fsap}, if $k \in \{ 1,2 \}$, and a $\frac{2k}{2k-1}$-approximation for any $k \geq 3$.
\end{theorem}

The proof of Theorem~\ref{thm:ratio_MaxSmall} uses the next lemma.
Recall that $G_1 \subseteq G'$ is the subgraph induced by the intervals
$J_i$ for which $|c(J_i)|<\MaxW$.
\begin{mylemma}
\label{lemma:G1_proper_two_colorable}
The subgraph $G_1 \subseteq G'$ is proper and $2$-colorable.
\end{mylemma}

\begin{proof}
We first note that if $G_1$ is not proper, then there exist two intervals, $J_i$ and $J_j$, such that
$J_j$ is properly contained in $J_i$. By the way Paging\_{\fba} proceeds, when it colors $J_j$, some colors
that were assigned to $J_i$ should be assigned to $J_j$, until either $|c(J_j)|=\MaxW$, or
$|c(J_i)|=0$. Since none of the two occurs - a contradiction.

We now show that $G_1$ is $2$-colorable. Throughout the proof, we assume that
there are $n_1$ intervals in $G_1$ sorted by their
starting points, and numbered
$1, 2, \ldots,n_1$ ,
i.e., $s_1 < s_2, < \cdot \cdot \cdot < s_{n_1}$.
For any $t \in [s_1,e_{n_1})$, say that $t$ is {\em tight} if
$\sum_{ \{J_\ell \in {\cJ}: t \in J_\ell \}} |c(J_\ell)| =W$.
Let $\cT$ denote the set of tight time points.
To keep a discrete set of such points, we consider only tight points $t$ which are also the start-times of
intervals, i.e., $t=s_i$ for some $1 \leq i \leq n$.
We note that every interval $J_i \in G_1$ contains at least one tight time
point (otherwise, Paging\_{\fba} would assign more colors to $J_i$). We complete the proof using the next claim.

\begin{cl}
\label{claim:tight_points_in_G1}
Every time point  $t \in \cT$ is contained in exactly one interval  $J_i \in G_1$.
\end{cl}

We now show that Claim \ref{claim:tight_points_in_G1} implies that $G_1$ is $2$-colorable.
Assume that there exists in $G_1$ a clique of at least $3$ intervals.
Let  $J_{a}, J_{b}, J_{c}$ be three ordered intervals in this clique.
We note that there is no $t \in J_b$ that is not contained in either
$J_a$ or $J_c$. However, $J_b$ must contain a tight time point.
Contradiction to Claim
\ref{claim:tight_points_in_G1}.

\noindent
{\bf Proof of  Claim \ref{claim:tight_points_in_G1}:}
\remove{
We show the two properties in the claim.
\begin{enumerate}
\item[(i)]
We use induction on $i$ to show that, for any $t \in T_i$, $J_i$ is alone in $G_1$ at $t$,
for all $J_i \in G_1$.
\\
{\bf Basis:} $i=1$. Let $t \in T_1$, then we distinguish between two cases.
\begin{enumerate}
\item[(a)]
If $t < s_2$ then we are done.
\item[(b)]
Otherwise, Paging\_{\fba} assigns to $J_2$ colors only if $|c(J_1)|= \MaxW$. Contradiction.
\end{enumerate}
{\bf Induction Step:} Consider a point $t \in T_i$ for $i >1$.
We note that, by the induction hypothesis, if $t \in T_j$ for some $j <i$, then $J_j$ is alone in $G_1$
at $t$, namely, $J_i \notin G_1$. Thus, we have two cases to consider.
\begin{enumerate}
\item[(a)]
If $e_{i-1} < t < s_{i+1}$ then we are done.
\item[(b)]
If $t > s_{i+1}$, then Paging\_{\fba} assigns colors to $J_{i+1}$ only $|c(J_i)|= \MaxW$. Contradiction.
\end{enumerate}
\item[(ii)]
By $(i)$, for any $J_i \in G_1$, consider a point $t \in T_i$. Then, since all of the intervals
that intersect $J_i$ at $t$ belong to $G_2$ (i.e., each of these intervals is assigned $\MaxW$
colors, it follows that $|c(J_i)|= W \bmod \MaxW$. \hfill \eod
\end{enumerate}
}
First, note that since $W$ is not a multiple of $\MaxW$ every tight
time point has to be contained in at least one interval $J_i \in G_1$.
We prove that it cannot be contained in more than one such interval by induction.
Let $t_1$ be the earliest time point in $\cT$. Since each $J_i \in
G_1$ contains a tight time point, $t_1 \in J_1$.
If $t_1 < s_2$ then clearly the claim holds for $t_1$.
Suppose that $t_1 \ge s_2$. In this case, since there is at least one
available color
in $[s_1,s_2)$, and since $e_2 > e_1$, Algorithm Paging\_{\fba} would
assign at least one additional color to $J_1$ rather than to
$J_2$. A contradiction.

For the inductive step, let $i>1$, and consider $t_i \in \cT$. Assume that the claim holds for
$t_{i-1} \in \cT$, and let $t_{i-1}  \in J_{\ell}$.
Since $t_{i-1}$ is not contained in any other interval in $G_1$, and
since it is tight, we must have $|c(J_{\ell})|=W \bmod \MaxW$.
To obtain a contradiction assume that $t_i$ is contained in more than
one interval in $G_1$. At least one of these intervals must be
$J_{\ell+1}$ (since $t_i > e_{\ell-1}$).
If $t_i$ is not contained in $J_\ell$ then clearly, by our assumption, $t_i$
must also be contained in $J_{\ell+2}$. However, even if $t_i \in J_\ell$,
since $|c(J_{\ell})|=W \bmod \MaxW$, in order for $t_i$ to be tight, it
must be contained also in $J_{\ell+2}$.
In this case, since there is at least one
available color
in $[s_{\ell+1},s_{\ell+2})$, and since $e_{\ell+2} > e_{\ell+1}, $ Algorithm Paging\_{\fba} would
assign at least one additional color to $J_{\ell+1}$, rather than to
$J_{\ell+2}$. A contradiction.
\hfill \eod
\end{proof}

We are ready to show the performance ratio of Algorithm ${\ca}_{MaxSmall}$.

\noindent
{\bf Proof of Theorem \ref{thm:ratio_MaxSmall}:}
Given the graphs $G'$ and $G_1, G_2$, as defined in Steps \ref{step:def_G1}, and \ref{step:def_G2} of
${\ca}_{MaxSmall}$, respectively. Let ${OPT}_{\fsapss}(G)$ and ${\ca}(G)$ be the value of an optimal
solution and the solution output by ${\ca}_{MaxSmall}$, respectively, for an input graph $G$.
Clearly, ${OPT}_{\fsapss}(G) \le {OPT}_{\fbass}(G)$.

\par\noindent{Case 1:} $k<3$. Since in this case $G_1$ is an independent set,
${\ca}(G) =  {OPT}_{\fbass}(G)$, and thus ${\ca}(G)=  {OPT}_{\fsapss}(G)$.

\par\noindent{Case 2:}
$k\ge 3$. Since in this case $G_1$ is 2-colorable,
${\ca}(G) \ge  {OPT}_{\fbass}(G) - \frac 12 \cdot (W \bmod \MaxW)\cdot |G_1|$.
Hence, to get the approximation ratio $\frac{2k}{2k-1}$, we need to show that
$(W \bmod \MaxW)\cdot |G_1| \le \frac {1}{k} {OPT}_{\fbass}(G)$.
Let $\lambda\MaxW = W \bmod\MaxW$. Note that $0 < \lambda < 1$.
In fact, we prove a slightly better bound, as
we show that
\begin{equation}\label{eqn:bound_G1}
\frac{(W \bmod \MaxW)\cdot |G_1|}{{OPT}_{\fbass}(G)} \le
\frac{\lambda}{k-1+\lambda} < \frac{1}{k}.
\end{equation}
This implies the approximation ratio
$\frac{2k-2+2\lambda}{2k-2+\lambda} < \frac{2k}{2k-1}$.
To obtain a contradiction, assume that inequality~(\ref{eqn:bound_G1})
does not hold. Recall that ${OPT}_{\fbass}(G) = (W \bmod \MaxW)\cdot
|G_1|+\MaxW\cdot |G_2|$, and that $|G_2| = |G'|-|G_1|$. We get
\begin{align*}
    \frac{(W \bmod \MaxW)\cdot |G_1|}{{OPT}_{\fbass}(G)} & =
    \frac{\lambda\MaxW|G_1|}{\MaxW|G'|-\MaxW(1-\lambda)|G_1|} &=
    \frac{\lambda|G_1|}{|G'|-(1-\lambda)|G_1|} &>
\frac{\lambda}{k-1+\lambda}.
\end{align*}
This implies $k|G_1| > |G'|$, and thus
\begin{align*}
    {OPT}_{\fbass}(G) &= \lambda\MaxW|G_1|+\MaxW(|G'|-|G_1|) =
    \MaxW|G'| - (1-\lambda)\MaxW|G_1| \\ &<
(1-\frac{1}{k}(1-\lambda))\MaxW|G'|=\frac{k-1+\lambda}{k}\cdot\MaxW|G'|.
\end{align*}

Note than $G'$, the support graph of the solution obtained by
Paging\_{\fba}, is $k$-colorable, since it cannot contain any clique of
size $k+1$. Indeed, such a clique can have at most $2$ intervals from $G_1$,
and at least $k-1$ intervals from $G_2$, and thus requires more than $W$
colors (since Algorithm Paging\_{\fba} assigns $W \bmod\MaxW$ colors to each interval in $G_1$).
We can use the $k$ coloring to obtain a solution of the
{\fbap} instance, by assigning $\MaxW$ colors to intervals in the $k-1$ largest
color classes, and $W \bmod\MaxW$ to the remaining color class.
Thus,
${OPT}_{\fbass}(G) \ge
\frac{1}{k}\cdot\lambda\MaxW|G'|+\frac{k-1}{k}\cdot\MaxW|G'|=\frac{k-1+\lambda}{k}\cdot\MaxW|G'|$.
A contradiction.
\remove{
We distinguish between two cases.
\begin{enumerate}
\item[(i)]
Suppose that $k \geq 3$. Then, we have
\[
\begin{array}{ll}
\displaystyle{
\frac{{OPT}_{\fsapss} (G)}{{\ca}(G)}} & \displaystyle {\leq \frac{{OPT}_{\fsapss} (G_1) + {OPT}_{\fsapss} (G_2)}
{\frac{{OPT}_{\fsapss} (G_1)}{2} + {OPT}_{\fsapss} (G_2)} }\\
\\
& \displaystyle{\leq \frac{{OPT}_{\fsapss} (G')}{{OPT}_{\fsapss} (G') (\frac{1}{2k} + \frac{k-1}{k})} = \frac{2k}{2k-1}.}
\end{array}
\]
The first inequality follows from the fact that ${\ca}_{MaxSmall}$ colors $G_2$ optimally and finds a
maximum independent set in $G_1$, which is $2$-colorable. The second inequality holds since $G_2$ is
a maximum $(k-1)$-colorable subgraph, and therefore, ${OPT}_{\fsapss}(G_2) \geq \frac{k-1}{k} {OPT}_{\fsapss}(G)$.
\item[(ii)]
If $k=2$, we note that $G_1$ is an independent set, thus ${\ca}_{MaxSmall}$ yields an optimal solution.
\end{enumerate}
}
\hfill \eod

\end{document}